\documentclass[reqno, 11pt]{amsart}
\usepackage{amssymb}%, natbib}
\usepackage[nesting]{hyperref}
\usepackage[pdftex]{graphicx}
\usepackage{listings}
\usepackage{multirow}
\usepackage{placeins}
\usepackage{color}
\usepackage{subfigure}
\usepackage{lscape}
\usepackage{dsfont}
\usepackage{tikz}
\usetikzlibrary{shapes,snakes}
\usepackage{verbatim}
\usetikzlibrary{arrows, decorations.markings}
\tikzstyle{vecArrow} = [thick, decoration={markings,mark=at position
   1 with {\arrow[semithick]{open triangle  60}}},
   preaction = {decorate}]%,
\newcommand{\BlackBoxes}{\global\overfullrule5pt}

\BlackBoxes

\newcommand{\R}{\mathbb{R}}

\newcommand{\Eop}{\mathbb{E}}
\newcommand{\Prob}{\mathbb{P}}

\newtheorem{theorem}{Theorem}
\newtheorem{corollary}[theorem]{Corollary}

\newtheorem{proposition}[theorem]{Proposition}
\theoremstyle{definition}
\newtheorem{example}[theorem]{Example}
\newtheorem{remark}[theorem]{Remark}
\newtheorem{definition}[theorem]{Definition}
\numberwithin{equation}{section} \numberwithin{theorem}{section}

\def\0{\kern0pt\-\nobreak\hskip0pt\relax}

\makeatletter
\AtBeginDocument{%
 \def\@serieslogo{%
 \vbox to\headheight{%
 \parindent\z@ \fontsize{6}{7\p@}\selectfont
 \vss}}}

\def\makeoverbar#1#2#3#4#5#6#7{%
 \setbox0=\hbox{$\m@th#2\mkern#5mu{{}#3{}}\mkern#6mu$}%
 \setbox1=\null \dimen@=#4\fontdimen8#13 \dimen@=3.5\dimen@
 \advance\dimen@ by \ht0 \dimen@=-#7\dimen@ \advance\dimen@ by \wd0
 \ht1=\ht0 \dp1=\dp0 \wd1=\dimen@
 \dimen@=\fontdimen8#13 \fontdimen8#13=#4\fontdimen8#13
 \rlap{\hbox to \wd0{$\m@th\hss#2{\overline{\box1}}\mkern#5mu$}}
 \fontdimen8#13=\dimen@}

\def\mylabel#1#2{{\def\@currentlabel{#2}\label{#1}}}

\makeatother

\allowdisplaybreaks

\begin{document}

%\listoftodos

\makeatletter \providecommand\@dotsep{5} \makeatother
%\listoftodos[Changes in Orange/Red To Do List in Green / Blue]\relax

\title[Optimal Risk Allocation in Reinsurance Networks]{Optimal Risk Allocation in Reinsurance Networks}

\author[N. \smash{B\"auerle}]{Nicole B\"auerle${}^*$}
\address[N. B\"auerle]{Department of Mathematics,
Karlsruhe Institute of Technology, D-76128 Karlsruhe, Germany}

\email{\href{mailto:nicole.baeuerle@kit.edu}
{nicole.baeuerle@kit.edu}}

\author[A. \smash{Glauner}]{Alexander Glauner${}^\dag$}
\address[A. Glauner]{Department of Mathematics,
Karlsruhe Institute of Technology, D-76128 Karlsruhe, Germany}

\email{\href{mailto:alexander.glauner@kit.edu} {alexander.glauner@kit.edu}}

%\thanks{${}^*$ Department of Mathematics, Karlsruhe Institute of Technology, D-76128 Karlsruhe, Germany }

\begin{abstract}
In this paper we consider reinsurance or risk sharing from a macroeconomic point of view.  Our aim is to find socially optimal reinsurance treaties.  In our setting we assume that there are $n$ insurance companies each bearing a certain risk and one representative reinsurer. The optimization problem  is to minimize the sum of all capital requirements of the insurers where we assume that all insurance companies use a form of Range-Value-at-Risk.  We show that in case all insurers use Value-at-Risk and the reinsurer's premium principle satisfies monotonicity, then  layer reinsurance treaties are  socially optimal. For this result we do not need any dependence structure between the risks. In the general setting with Range-Value-at-Risk we obtain again the optimality of layer reinsurance treaties under further assumptions, in particular under the assumption  that the individual risks are positively dependent through the stochastic ordering. Our results include the findings in \cite{CT13} in the special case $n=1$. At the end,  we discuss the difference between socially optimal reinsurance treaties and individually optimal ones by looking at a number of special cases.
\end{abstract}
\maketitle

\vspace{0.5cm}
\begin{minipage}{14cm}
{\small
\begin{description}
\item[\rm \textsc{ Key words}  ]
{\small Optimal reinsurance, Range-Value-at-Risk, Positively dependent through stochastic ordering }
%\item[\rm \textsc{AMS subject classifications}]{\small Primary ?? Secondary }
\end{description}
}
\end{minipage}

\section{Introduction and Motivation}\label{sec:intro}
Finding the optimal form of a reinsurance treaty or in more general terms,  optimizing risk sharing, is an old topic which regained a lot of attention in recent years. One of the first starting points has been \cite{borch60} who proved that a stop-loss reinsurance treaty minimizes the retained loss of the insurer given the reinsurance premium is calculated with the expected value principle. A similar result has been derived in \cite{Arrow63} where the expected utility of terminal wealth of the insurer has been maximized. Since then a lot of generalizations of this problem have been considered. We refer the interested reader to the recent book \cite{ABT17} which contains a comprehensive literature overview in chapter 8 and to \cite{DP14}. We will here only mention a few recent articles which are relevant for our study. First in \cite{BBH09}  a characterization of optimal reinsurance forms for a general class of risk measures has been given by exploiting duality theory in functional analysis. A stop-loss  treaty turned out to be optimal when the premium principle is an expected value principle. Further \cite{CT13} considered the optimization problem with Value-at-Risk and Expected Shortfall and a general premium principle for the insurer. They obtain the optimality of a layer-reinsurance.

While most publications consider the problem only from the perspective of the individual insurer, we investigate the situation from an economic point of view. More precisely, we want to know what kind of risk sharing between  insurers and reinsurer is  optimal for the entire economy and in which situations is it identical to the individually optimal decision of the insurer? This question also makes it necessary to address the task of modelling the problem for a random vector representing the individual risks taken by the insurers. There exists of course a rich literature on risk sharing problems where random vectors are involved. The most popular problem is the so-called {\em inf-convolution} problem which is given by
$$ \min \; \sum_{i=1}^n \rho_i(X_i) \quad s.t. \;\; X_1+\ldots + X_n=X$$
where $\rho_i$ are suitable risk measures.  It has been shown in \cite{FS08} that for law- and cash-invariant convex risk measures a solution always exists and is given by a comonotone structure.  This result has been refined by \cite{ELW16} where it has been shown that if the risk measures are given by Range-Value-at-Risk, there is an explicit construction for the optimal solution. In \cite{KR13} this problem has been interpreted in a setting with several insurers with general convex risk measures and premium principles. There, optimal reinsurance contracts have been characterized by means of subdifferential formulas in Banach spaces. For more results on the inf-convolution problem we refer the reader to \cite{Ru13}.

Problems where special kinds of  risk sharing between two entities are considered can be found e.g. in \cite{ABT13}. There, the insurance group allocates the total risk between two entities which are subject to different regulatory capital requirements, using appropriate risk transfer agreements. The optimal risk sharing rule is derived explicitly for special risk measures like Value-at-Risk and Expected Shortfall. In \cite{CLL16} the authors develop optimal reinsurance contracts that minimize the convex combination of the Value-at-Risk of the insurer's loss and the reinsurer's loss under some constraints. Some explicit, though rather complicated optimal reinsurance treaties are obtained there. Next, \cite{CM14} investigate the optimal form of reinsurance from the perspective of an insurer when he decides to cede parts of the loss to two reinsurers, where both reinsurers calculate the premium according to different premium principles. The problem is solved under the criterion of minimizing Value-at-Risk  or Expected Shortfall. An  optimal reinsurance treaty is to cede two adjacent layers of the risk.
Another multivariate problem is considered in  \cite{ZCW14} where optimal reinsurance strategies for an insurer with multiple lines of business  are investigated under the criterion of minimizing the total capital requirement calculated based on the multivariate lower-orthant Value-at-Risk. The optimal strategy for the insurer there is to buy a two-layer reinsurance treaty for each line of business.  Note that the dependence structure for the individual risks was not important for the results cited so far. A worst case scenario w.r.t.\ the dependence structure has been considered in \cite{CSY14} where  the problem of optimal reinsurance treaties for
multivariate risks with general law-invariant convex risk measures has been studied. It turned out that stop-loss reinsurance treaties minimize a general law-invariant convex risk measure of the total retained
risk. In \cite{CW12} it has been assumed that  an insurer has $n$ lines of business which can be reinsured subject to a given premium and the aim is to minimize the expected convex function of the retained total risk. In order to derive results in this setting the authors needed a concept for positive dependence between risks which has been the concept of 'positively dependent through the stochastic ordering'.

Papers with a more economic point of view on optimal reinsurance are among others \cite{UL85} where a  Stackelberg equilibrium for $n$  reinsurers under special assumptions is considered and \cite{PS01} where a game-theoretic analysis of optimal insurance networks has been conducted.

The aim of this paper now is to consider reinsurance or risk sharing from a macroeconomic point of view. Whereas the individual goal of an insurance company is to reduce risk exposure and own capital requirements by reinsurance, the social goal of reinsurance is to spread risk around the globe by avoiding local overexposures. This construction also increases the amount of risk which can be insured.  In our setting we assume that there are $n$ insurance companies, each bearing a certain risk, and one representative reinsurer. In contrast to the inf-convolution problem the situation is no longer symmetric. The optimization problem then is to minimize the sum of all capital requirements of the insurers. We assume that all insurance companies use Range-Value-at-Risk as a risk measure with possibly different parameters.  Range-Value-at-Risk comprises Value-at-Risk and Expected Shortfall and is thus a natural choice with practical relevance. We show that in case all insurers use Value-at-Risk and the reinsurer's premium principle satisfies monotonicity, then  layer reinsurance treaties are  socially optimal. For this result we do not need any dependence structure between the risks. In the general setting with Range-Value-at-Risk we obtain again the optimality of layer reinsurance treaties under the assumption that   the reinsurer's premium principle is consistent with the increasing convex order (which most premium principles are) and under the assumption  that the individual risks are positively dependent through the stochastic ordering (PDS). Our results include the findings in \cite{CT13} in the special case $n=1$. Finally, we also discuss the difference between socially optimal reinsurance treaties and individually optimal ones. Fortunately, they coincide in many cases but there also may be some differences.

Our paper is organized as follows: In the next section we summarize some definitions and facts from risk measures, stochastic orders and dependence concepts. In particular we prove that  PDS random vectors carry the increasing convex order of the margins over to the sum of the components. In Section \ref{sec:optpro} we introduce and discuss our optimization problem. The solution of the problem is then presented in Section  \ref{sec:solution} where also some special cases are discussed. In the last section we investigate the difference between socially optimal reinsurance treaties and individually optimal ones by looking at a number of special cases.

\section{Risk Measures, Stochastic Orders and Dependence Structures}
We will consider non-negative random variables $X: \Omega \to \R_+$ defined on a non-atomic probability space $(\Omega, \mathcal{A},\Prob)$. They represent future insurance claims, i.e. $X(\omega) \geq 0$ is the discounted net loss of an insurance company at the end of a fixed period due to a policy (or portfolio of policies) sold by them.  We denote the (cumulative) distribution function by $F_X(x):=\Prob(X \leq x)$, the survival function by $S_X(x):=1-F_X(x)$ and the generalized inverse by $ F^{-1}_X(\alpha):= \inf\{x \in \R: \ F_X(x) \geq \alpha\}$ where $x \in \R$ and $\alpha \in [0,1]$. With
\[L^1:=\{X: \Omega \to \R_+: \ X \text{ is a random variable with } \Eop[X] < \infty \} \]
we denote the space of all such non-negative, integrable random variables. We now recall some notions of risk measures. In general, a risk measure is a mapping $\rho : L^1 \to \bar{\R}$. Essentially, the notion of a premium principle $\pi: L^1 \to  \bar{\R}$ is mathematically equivalent but applications are different. While the former determines the necessary solvency capital to bear a risk, the latter gives the price of (re)insuring it. The properties of risk measures discussed in the sequel apply to premium principles analogously. Of particular importance are the following risk measures.

\begin{definition}
For $\alpha,\beta \in [0,1]$ and $X \in L^1$ with distribution function $F_X$ we define 
\begin{itemize}
\item[a)] the \emph{Value-at-Risk} of $X$ at level $\alpha$ as $VaR_{\alpha}(X) := F^{-1}_X(1-\alpha)$.
\item[b)] the \emph{Expected Shortfall} of $X$ at level $\beta>0$ as
	$ES_{\beta}(X):= \frac{1}{\beta}\int_0^{\beta} VaR_{s}(X) ds.$
\item[c)]	 	  the \emph{Range-Value-at-Risk} of $X$ at level $\alpha, \beta$ if  $\alpha + \beta \leq 1$ as
	$$RVaR_{\alpha, \beta}(X):=  \begin{cases}
			\frac{1}{\beta} \int_{\alpha}^{\alpha + \beta} VaR_{s}(X) d s, & \beta >0\\
			VaR_{\alpha}(X), & \beta=0.
	\end{cases}$$	
\end{itemize}
\end{definition}
Obviously, Range-Value-at-Risk comprises both Value-at-Risk and Expected Shortfall.\\
A risk measure $\rho$ should have some nice properties like  for example
\begin{itemize}
\item[i)] {\em law-invariance:} $\rho(X)$ depends only on the distribution $F_X$.
\item[ii)] {\em monotonicity:} If $X\le Y$ then $\rho(X)\le \rho(Y)$.
\item[iii)] {\em translation invariance:} For $m\in\R$ it holds $\rho(X+m)=\rho(X)+m$.
\item[iv)] {\em positive homogeneity:} For $\alpha\ge 0$ it holds that $\rho(\alpha X)=\alpha\rho(X).$
\item[v)] {\em subadditivity:} $\rho(X+Y)\le \rho(X)+\rho(Y)$.
\item[vi)] {\em convexity:} For $\alpha\in [0,1]$ it holds that $\rho(\alpha X+(1-\alpha)Y)\le \alpha\rho(X)+(1-\alpha)\rho(Y).$
\end{itemize}

Though Value-at-Risk is in general not subadditive it has a lot of nice properties like law-invariance, monotonicity, translation invariance and positive homogeneity. These properties are also shared by Range-Value-at-Risk. The following facts, which can be directly derived from the definition, will be important for us:

\begin{proposition}\label{prop:varpro}
For $\alpha \in [0,1]$ and $X \in L^1$ we obtain:
\begin{itemize}
\item[a)] $\alpha \to VaR_{\alpha}(X)$ is decreasing in $\alpha$.
\item[b)] For any non-decreasing, left-continuous function $t: \R \to \R$ it holds $	VaR_{\alpha}(t(X)) = t\big(VaR_{\alpha}(X)\big).$
\end{itemize}
\end{proposition}

For the solution of our optimization problem we make use of the following notions of stochastic orderings.

\begin{definition}
	Let $X,Y$ be two  random variables. Then $X$ is less than $Y$ in
	\begin{itemize}
		\item[(a)] \emph{usual stochastic order} ($X \leq_{st}Y$) if $ \Eop[f(X)] \leq \Eop[f(Y)]$ for all increasing $f:\R\to\R$,
		\item[(b)] \emph{convex order} ($X \leq_{cx}Y)$ if $ \Eop[f(X)] \leq \Eop[f(Y)]$ for all convex $f:\R\to\R$,
		\item[(c)] \emph{increasing convex order} ($X \leq_{icx}Y)$ if $ \Eop[f(X)] \leq \Eop[f(Y)]$ for all increasing convex $f:\R\to\R$,
	\end{itemize}
	whenever the expectations exist.
\end{definition}

The following characterizations of the increasing convex order are well-known (see e.g. \cite{ms02}, chap. 1)
\begin{proposition}\label{thm:stop-loss}
	Let $X,Y$ be two random variables with distribution functions $F_X, F_Y$.
\begin{itemize}
\item[a)]  $X \leq_{icx} Y$ holds if and only if $\Eop[(X-t)]_+ \leq \Eop[(Y-t)]_+$ for all $t \in \R$.
\item[b)] $X \leq_{icx} Y$ holds if and only if  there exists another random variable $Z$ with  $X \leq_{st} Z \leq_{cx} Y. $
\item[c)] Let $t_0 \in \R$ such that $F_X(t) \leq F_Y(t)$ for $t < t_0$, $F_X(t) \geq F_Y(t)$ for $t \geq t_0$ and $\Eop[X] \leq \Eop[Y]$ then $X \leq_{icx} Y$.
\end{itemize}
\end{proposition}

From now on we will restrict to law-invariant risk measures.

\begin{definition} We say that a law-invariant risk measure $\rho$ is consistent 
\begin{itemize}
\item[a)] with the usual stochastic order if $X \leq_{st}Y$ implies $\rho(X) \leq \rho(Y)$.
\item[a)] with the increasing convex order if $X \leq_{icx}Y$ implies $\rho(X) \leq \rho(Y)$.
\end{itemize}
\end{definition}

It is well-known (see e.g.\ \cite{bm06}, Theorem 4.2) that on non-atomic probability spaces monotonicity of the risk measure is enough to imply consistency w.r.t.\ the usual stochastic ordering and the additional  property of convexity of the risk measure is enough to imply consistency w.r.t.\ the increasing convex ordering (see e.g.\ \cite{bm06}, Theorem 4.4). Hence a large class of risk measures satisfy these consistency properties. We also have to recall the concept of a copula for random vectors.

\begin{definition}\label{def:copula}
	Let $X=(X_1, \dots, X_n)$ be a random vector with distribution function $F$ and marginal distribution functions $F_1, \dots, F_n$. An n-dimensional distribution function $C$ with uniform marginals on $[0,1]$ is called a \emph{copula} of $X$ if
	\begin{align*}
		F(x_1, \dots, x_n) = C \left( F_1(x_1), \dots, F_n(x_n) \right), \quad (x_1, \dots, x_n) \in \R^n.
	\end{align*} 
\end{definition}
One important example of a copula is the so-called Fr\'echet-Hoeffding bound which is given by $$C(x_1,\ldots,x_n) = \min\{x_1,\ldots,x_n\}. $$
The following observation will be crucial for us. It is implied by the definition of the copula.

\begin{corollary}\label{thm:common_copula}
	Let $X=(X_1,\dots, X_n)$ be a random vector with copula $C$ and $t_1,\dots, t_n$ be increasing and continuous functions. Then the random vector 
	\[\left(t_1(X_1), \dots, t_n(X_n) \right)\]
	also has copula $C$.
\end{corollary}

A random vector which possesses as copula the Fr\'echet-Hoeffding bound is called {\em comonotonic}.

\begin{definition}
A risk measure $\rho$ is called {\em comonotone additive} if for all random vectors $X=(X_1, \dots, X_n)$ which are comonotonic we have that $$\rho\Big(\sum_{i=1}^n X_i\Big)  = \sum_{i=1}^n \rho(X_i). $$
\end{definition} 

The next result follows directly from \cite{ms01}, Theorem 4.1.

\begin{theorem}\label{thm:common_copula_st}
	Let $X=(X_1, \dots, X_n)$ and $Y=(Y_1, \dots , Y_n)$ be two random vectors with a common copula $C$. Then $X_i \leq_{st} Y_i, \ i=1, \dots, n$ implies $\sum_{i=1}^{n} X_i \leq_{st} \sum_{i=1}^{n} Y_i$.
\end{theorem}

Finally we need a concept of positive dependence. 

\begin{definition}\label{def:copula}
	The random vector $X=(X_1, \dots, X_n)$ is said to be
	\begin{itemize}
		\item[(a)] \emph{positively dependent through the stochastic ordering} (PDS) if for all $i=1,\ldots,n$
		$$\Eop[f(X_1,\ldots, X_{i-1},X_{i+1},\ldots,X_n)| X_i=x_i]$$
		is increasing in $ x_{i}$ for any (componentwise) increasing function $f: \R^{n-1} \to \R$.
		\item[(b)] \emph{positive orthant dependent} (POD) if both
		\begin{align*}
			\Prob(X \leq x) \geq \prod_{i=1}^{n} \Prob(X_i \leq x_i) \qquad \text{and} \qquad \Prob(X > x) \geq \prod_{i=1}^{n} \Prob(X_i > x_i)
		\end{align*}
		holds for all $x \in \R^n$.
	\end{itemize}
\end{definition}

Obviously, PDS and POD are properties of the copula. PDS is less well-known than other dependence concepts like e.g.\ conditionally increasing (CI) or $MTP_2$.  In \cite{bss} Theorem 5.3 it is shown that CI implies PDS. For $n=2$ the two conditions in Definition \ref{def:copula} (b) are equivalent. If for $n>2$ only the first or the second one holds, $X$ is referred to as positive lower or upper orthant depended respectively. PDS implies POD, cf.\ \cite{bss} Theorem 5.1. To the best of our knowledge, POD is even the smallest established superclass of PDS.

We need the following Theorem.

\begin{theorem}\label{thm:common_copula_icx}
	Let $X=(X_1, \dots, X_n)$ and $Y=(Y_1, \dots , Y_n)$ be two random vectors with a common PDS copula $C$. Then $X_i \leq_{icx} Y_i, \ i=1, \dots, n$ implies $\sum_{i=1}^{n} X_i \leq_{icx} \sum_{i=1}^{n} Y_i$.
\end{theorem}

\begin{proof}
According to Proposition \ref{thm:stop-loss} b) there exist random variables $Z_i$ such that $X_i \leq_{st} Z_i \leq_{cx} Y_i$ for $ i=1, \dots, n$. Let them be the marginals of a random vector $Z=(Z_1, \dots, Z_n)$ with copula $C$.  From Theorem \ref{thm:common_copula_st} it follows that $\sum_{i=1}^{n} X_i \leq_{st} \sum_{i=1}^{n} Z_i$ and from Corollary 3.5 in \cite{CW12} it follows that $\sum_{i=1}^{n} Z_i \leq_{cx} \sum_{i=1}^{n} Y_i$. Hence the statement is again obtained with Proposition \ref{thm:stop-loss} b).
\end{proof}

A common PDS copula is the mildest possible assumption to obtain this result. The next example shows that weakening the dependence concept to POD is not possible. It is a simplified version of Example 4.7 in \cite{ms01}, who used it in a slightly different context.

\begin{example}
	Let $n=2$. We define the distribution of the random vector $Y=(Y_1,Y_2)$ by its discrete density
	\[
		\begin{array}{r||c|c|c|c}
			\Prob(Y_1=y_1,Y_2=y_2) & y_2=0 & y_2=1 & y_2=3 & y_2=4  \\ \hline \hline
			\rule{0pt}{13pt} y_1=0 & \frac{3}{12} & 0 & \frac{2}{12} & \frac{1}{12} \\ \hline
			\rule{0pt}{13pt} y_1=1 & \frac{1}{12} & \frac{2}{12} & 0 & \frac{3}{12}
		\end{array}
	\]
	It is readily checked that $Y$ is POD. However, it is not PDS since
	\[\Prob(Y_1>0| Y_2=1) = 1 > 0 = \Prob(Y_1>0 | Y_2 = 3).\]
	Let us define an increasing function $f$ by $f(0)=0$, $f(1)=f(3)=2$, $f(4)=4$ and set $X=(Y_1, f(Y_2))$. Then, $X$ and $Y$ have the same copula, cf.\ Corollary \ref{thm:common_copula}. Furthermore, it follows from Proposition \ref{thm:stop-loss} c) that $X_i \leq_{icx} Y_i$, $i=1,2$ but $Y_1 + Y_2 \leq_{icx} X_1+X_2$.
\end{example} 

\section{The Optimization Problem}\label{sec:optpro}
We consider an economy with $n$ insurance companies, numbered consecutively $i=1, \dots, n$, and a single reinsurance company. The single reinsurer is justified by the study in \cite{BTZ16}. It is assumed that  insurer $i$ bears a total risk modelled by a non-negative random variable $X_i \in L^1$. These risks are interpreted as discounted losses at the end of a fixed period due to insurance claims based on policies sold by the respective insurer. Note that these risks are not necessarily independent. Each insurer $i$ evaluates its risk with a translation invariant and positive homogeneous risk measure $\rho_i$ which will be specified later on. 

In order to reduce the risk borne autonomously, insurance company $i$ may cede a portion $f_i(X_i)$ to the reinsurer. The retained risk is then given by $R_{f_i}(X_i) = X_i-f_i(X_i)$. We assume that  $f_1, \dots, f_n: \R_+ \to \R_+$ are increasing and that $f_i(x) \leq x$ for all $x \in \R_+$ and all $i=1, \dots,n$. Moreover, in order to rule out moral hazard, the retained loss functions $R_{f_1}, \dots, R_{f_n}: \R_+ \to \R_+$  are assumed to be increasing as well.  We define the set of admissible ceded loss functions by
\[\mathcal{C} = \{f:\R_+ \to \R_+ | \ f(x) \leq x \ \forall x \in \R_+ \text{ and } f, R_f \text{ are increasing}\}.  \]
Note that functions in $\mathcal{C}$ are in particular Lipschitz-continuous, since $R_f$ increasing leads to $f(x_2)-f(x_1) \le x_2-x_1$ for all $0\le x_1\le x_2$.
Given all ceded risks, the reinsurer prices their sum with a premium principle $\pi:L^1 \to \bar{\R}$ to take account of the dependence structure of the individual risks. The reinsurer then determines the amount to be paid by each insurer with a premium allocation rule.
\begin{definition}\label{def:premium_allocation}
	For a given premium principle $\pi$ and aggregate ceded risk $Y,$ a linear functional $\psi_{\pi}(\cdot, Y): L^1 \to \bar \R$ with $\psi_{\pi}(Y, Y) = \pi(Y)$ is called \emph{premium allocation rule}. 
\end{definition}
That is, insurer $i$ has to pay the reinsurance premium $\psi_{\pi}\left(f_i(X_i), \sum_{j=1}^{n} f_j(X_j) \right)$. The setting is illustrated in Figure \ref{fig:ORP}.

\begin{figure}[h]
	\centering
	\begin{tikzpicture}
 \node[inner sep=0,minimum size=0] (z) {};	 \node[draw,circle,right of=z] (a) {$X_1$};
  \node[inner sep=0,minimum size=0,right of=a] (k) {}; % invisible node
  \node[draw,circle,right of=k] (b) {$X_2$};
  \node[draw,circle,below of=z] (c) {$X_6$};
   \node[inner sep=0,minimum size=0,right of=c] (l) {}; 
   \node[draw,circle,below of=k] (re) {RI};
    \node[inner sep=0,minimum size=0,right of=re] (m) {}; 
   \node[draw,circle,right of=m] (d) {$X_3$};
    \node[draw,circle,below of=l] (e) {$X_5$};
     \node[draw,circle,below of=m] (f) {$X_4$};
\draw[vecArrow] (a) to (re);
\draw[vecArrow] (b) to (re);
\draw[vecArrow] (c) to (re);
\draw[vecArrow] (d) to (re);
\draw[vecArrow] (e) to (re);
\draw[vecArrow] (f) to (re);

	\end{tikzpicture}
	\vspace{0.5cm}
	\caption{The optimal reinsurance problem in an economy with one reinsurer and $n$ insurers. Insurer $i$ cedes $f_i(X_i)$ of his original risk $X_i$.}\label{fig:ORP}
\end{figure}
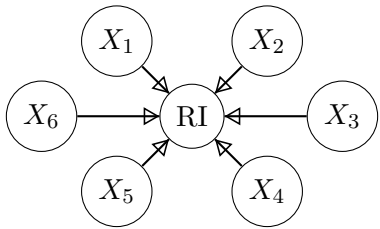

We are interested in a risk allocation which is macroeconomically optimal, i.e.\ which minimizes the aggregate capital requirement of all insurers. An insurance company's solvency (or risk) capital provides protection against insolvency due to large unexpected losses. Given that all capital requirements are representative for the borne risks, such an optimal allocation can be seen as desirable state of the economy with respect to financial stability. While $X_i$ models insurer $i$'s discounted loss due to insurance claims, the discounted net loss after reinsurance at the end of the fixed period is given by the retained loss less the insurer's premium income $\pi_i(X_i)$ in the period plus cost of reinsurance and cost of capital:
\begin{align}\label{eq:net_loss}
	Loss_i = R_{f_i}(X_i)-\pi_i(X_i) + \psi_{\pi}\Big(f_i(X_i), \sum_{j=1}^{n} f_j(X_j) \Big) + r_{CoC} \cdot
	 \rho_i(Loss_i).
\end{align}
Cost of capital is defined as some cost of capital rate $r_{CoC}$ times the capital requirement $\rho_i(Loss_i)$. We assume a uniform cost of capital rate $r_{CoC}$ for all insurance companies. The implicit description of the net loss (\ref{eq:net_loss}) goes back to \cite{k09}. Applying $\rho_i$ on both sides and using translation invariance and positive homogeneity, one obtains as capital requirement of insurer $i$
\[\rho_i(Loss_i) = \frac{1}{1-r_{CoC}} \left( \rho_i\big(R_{f_i}(X_i)\big) - \pi_i(X_i) + \psi_{\pi}\Big(f_i(X_i), \sum_{j=1}^{n} f_j(X_j) \Big)  \right). \]
Hence, the minimization objective becomes
\begin{equation}\label{eq:pre-ORP}
	\min\quad \sum_{i=1}^{n} \rho_i(Loss_i) \quad  s.t.\quad  f_1,\ldots ,f_n \in \mathcal{C}.
\end{equation}

As a strictly monotone transformation of the objective function it is sufficient to consider the sum of the capital requirements for the retained risks plus the total reinsurance premium
\[ \sum_{i=1}^{n} \left(\rho_i\big(R_{f_i}(X_i)\big) + \psi_{\pi}\Big(f_i(X_i), \sum_{j=1}^{n} f_j(X_j) \Big) \right) =  \sum_{i=1}^{n} \rho_i\big(R_{f_i}(X_i)\big) + \pi \left( \sum_{i=1}^{n} f_i(X_i) \right). \]
Therefore, we will study the optimal reinsurance problem 
\begin{equation}\label{eq:ORP}
\min\quad \sum_{i=1}^n \rho_i \big( R_{f_i}(X_i) \big) + \pi \left( \sum_{i=1}^{n} f_i(X_i) \right) \quad  s.t.\quad  f_1,\ldots ,f_n \in \mathcal{C}.
\end{equation}
In what follows we assume that the insurance companies use Range-Value-at-Risk as risk measure, i.e. for $i=1, \dots,n$ it holds $\rho_i=RVaR_{\alpha_i,\beta_i}$, where  $\alpha_i , \beta_i \geq 0$ such that $0 < \alpha_i + \beta_i \leq 1$.
 
As a superclass of Value-at-Risk and Expected Shortfall, the two most widely used risk measures in the financial industry, Range-Value-at-Risk is a natural choice to obtain results which cover most cases of practical relevance. Note that we require the sum of the parameters to be strictly positive to avoid the case of an infinite capital requirement.

\section{Reduction to an Optimization Problem with finite Dimension}\label{sec:solution}
\subsection{Case 1: Insurers use Value-at-Risk}\label{sec:VaR-min}
In this section we consider the special case that all insurers use Value-at-Risk, that is, $\beta_1= \dots=\beta_n=0$. So let $\rho_i= VaR_{\alpha_i}$ with $\alpha_i \in (0,1]$. Moreover, for the premium principle of the reinsurer we assume

\begin{description}
	\item[(A1)] The reinsurer's premium principle is consistent with the usual stochastic order, that is, $\pi(X) \leq \pi(Y)$ whenever $X \leq_{st} Y$ for $X,Y \in L^1$.
\end{description}

For reducing (\ref{eq:ORP}) to finite dimension, the following class of reinsurance treaties (i.e.\ ceded loss functions) plays a key role.
\begin{definition}\label{def:layer}
	The function $f: \R_+ \to \R_+$ given by
	\[f(x) = \min\{(x-a)_+,b\} = (x-a)_+ - (x-(a+b))_+, \quad a,b \geq 0\]
	is called \emph{layer reinsurance treaty} with deductible $a$ and upper bound $b$. In the degenerate case $b = \infty$ 
	\[f(x)=(x-a)_+ \]
	is referred to as \emph{stop-loss reinsurance treaty}.
\end{definition}
Note that all such functions belong to $\mathcal{C}$. Economically, layer reinsurance means that the reinsurer covers all losses exceeding the deductible but limits its liability to a maximum  $b$. The idea how to reduce (\ref{eq:ORP}) to finite dimension is to construct  layer reinsurances which are at least as good as  given, arbitrary ceded loss functions with respect to the objective function value. 

For each insurer $i=1, \dots, n$ let a ceded loss function $f_i \in \mathcal{C}$ be given and define $h_{f_i}: \R_+ \to \R_+$ by
\begin{align}
	h_{f_i}(x) = \min \Big\{ \Big(x - \big(VaR_{\alpha_i}(X_i) - f_i(VaR_{\alpha_i}(X_i)) \big)\Big)_+, f_i(VaR_{\alpha_i}(X_i)) \Big\rbrace . \label{eq:h_f}
\end{align}
These are layer reinsurance treaties with deductibles $\big(VaR_{\alpha_i}(X_i) - f_i(VaR_{\alpha_i}(X_i)) \big)$ and upper bounds $f_i(VaR_{\alpha_i}(X_i))$. It holds $h_{f_i} \in \mathcal{C}$ for $i=1, \dots, n$.

\begin{proposition} \label{thm:VaR_leq}
	Assume (A1) and let $f_1, \dots, f_n \in \mathcal{C}$ be arbitrary ceded loss functions. Then
	\begin{align*}
		\sum_{i=1}^n VaR_{\alpha_i} \Big( R_{h_{f_i}}(X_i) \Big) + \pi \Big( \sum_{i=1}^{n} h_{f_i}(X_i) \Big) \leq \sum_{i=1}^n VaR_{\alpha_i} \Big( R_{f_i}(X_i) \Big) + \pi \Big( \sum_{i=1}^{n} f_i(X_i) \Big).
	\end{align*}
\end{proposition}

\begin{proof}
	Let $i \in \{1, \dots, n\}$ be arbitrary. By inserting we get $h_{f_i}(VaR_{\alpha_i}(X_i)) = f_i(VaR_{\alpha_i}(X_i))$. The functions $h_{f_i}, f_i$ are in $\mathcal{C}$ and therefore $R_{h_{f_i}}, R_{f_i}$ are increasing and continuous. Thus, the properties of $VaR_{\alpha_i} $ (see Proposition \ref{prop:varpro}) imply
	\begin{align}
		VaR_{\alpha_i} \Big(R_{h_{f_i}}(X_i) \Big) &= R_{h_{f_i}}\big( VaR_{\alpha_i}(X_i) \big) \notag \\ 
		&= VaR_{\alpha_i}(X_i)- h_{f_i}\big(VaR_{\alpha_i}(X_i) \big) \notag \\
		& = VaR_{\alpha_i}(X_i)- f_i\big(VaR_{\alpha_i}(X_i) \big) \notag \\
		&= R_{f_i}\big( VaR_{\alpha_i}(X_i) \big) \notag \\
		&= VaR_{\alpha_i} \Big(R_{f_i}(X_i) \Big). \label{eq:R_f_VaR}
	\end{align}
	Now, we show $h_{f_i}(x) \leq f_i(x)$  for all  $x \in \R_+$. Since $R_{f_i}$ is increasing it follows for all  $x \leq VaR_{\alpha_i}(X_i)$ that
	\[f_i\big(VaR_{\alpha_i}(X_i)\big) - f_i(x)  \leq VaR_{\alpha_i}(X_i) - x. \]
	Rearranging and using that $f_i \geq 0$ gives
	\[f_i(x) \geq \left( x - VaR_{\alpha_i}(X_i) + f_i\big(VaR_{\alpha_i}(X_i)\big)\right)_+ = h_{f_i}(x),\]
where the equality holds due to $x \leq VaR_{\alpha_i}(X_i)$. For $x \geq VaR_{\alpha_i}(X_i)$ we have
\[h_{f_i}(x) = f_i\big(VaR_{\alpha_i}(X_i)\big) \leq f_i(x), \]
since $f_i$ is increasing. Hence, we have shown $h_{f_i}(x) \leq f_i(x)$  for all  $x \in \R_+$ which in turn implies
\begin{align}
	h_{f_i}(X_i) \leq_{st} f_i(X_i). \label{eq:h_f-st-f}
\end{align}
 As $i$ was arbitrary, Theorem \ref{thm:common_copula_st} yields
\[ \sum_{i=1}^{n} h_{f_i}(X_i) \leq_{st} \sum_{i=1}^{n} f_i(X_i).\]
By (A1) it follows
\begin{align}
\pi\Big( \sum_{i=1}^{n} h_{f_i}(X_i) \Big) \leq \pi\Big(\sum_{i=1}^{n} f_i(X_i) \Big). \label{eq:pi1}
\end{align}
(\ref{eq:R_f_VaR}) and (\ref{eq:pi1}) together yield the assertion.
\end{proof}	

\begin{theorem} \label{thm:reduction_VaR}
	Let all insurers use Value-at-Risk with parameters $\alpha_1, \dots, \alpha_n$ and assume (A1). Regardless of the dependence structure between their individual risks $X_1, \dots, X_n$  the optimal reinsurance problem (\ref{eq:ORP}) has the same optimal value as the finite dimensional problem
	\begin{align}\label{eq:FORP_1}
	\begin{aligned}
{\min} \quad  \sum_{i=1}^n a_i + \pi \Big( \sum_{i=1}^{n} \min\{(X_i-a_i)_+, VaR_{\alpha_i}(X_i) - a_i\}\Big) \quad s.t. \quad  0 \leq a_i \leq VaR_{\alpha_i}(X_i), \forall  i.
	\end{aligned}
	\end{align}
	Furthermore, if $a_1, \dots, a_n$ is an optimal solution to (\ref{eq:FORP_1}) then
	\[f_i(x) = \min\{(x-a_i)_+, VaR_{\alpha_i}(X_i) - a_i\}, \quad i=1, \dots,n \] is an optimal solution to (\ref{eq:ORP}).
\end{theorem}

\begin{proof}
	Let $i \in \{1, \dots, n\}$,  $a_i \leq VaR_{\alpha_i}(X_i)$ and 
\[f_i(x) = \min\{(x-a_i)_+, VaR_{\alpha_i}(X_i) - a_i\}. \] 	
	Then $f_i \in \mathcal{C}$ and since $VaR_{\alpha_i} \big(R_{f_i}(X_i) \big)= R_{f_i}\big(VaR_{\alpha_i}(X_i)\big)=a_i$ the target function in (\ref{eq:ORP}) reduces to \eqref{eq:FORP_1}. Hence we have that the minimum value of \eqref{eq:FORP_1} is at least as large as that of \eqref{eq:ORP}. On the other hand, Proposition \ref{thm:VaR_leq} yields that the minimum value of \eqref{eq:ORP} is at least as large as that of \eqref{eq:FORP_1}. Hence, the optimal values are equal and the statement follows.
\end{proof}

\begin{remark}
Note that in the setting of Theorem \ref{thm:reduction_VaR} an optimal solution always exists when $\pi$ is Lipschitz continuous w.r.t.\ the supremum norm. This is due to the fact that the target function is then Lipschitz-continuous w.r.t.\ any norm of $\R^n$ and  minimized over a compact set. The existence of a factor $\lambda>0$ such that $\lambda \pi$ becomes translation invariant is a sufficient condition for the Lipschitz-continuity of $\pi$. Moreover, note that the random variables which appear as an argument of $\pi$ are bounded. 
\end{remark}

\subsection{Case 2: Insurers use a general Range-Value-at-Risk}
This section covers the general case when at least one insurance company uses Range-Value-at-Risk with a positive parameter $\beta$. Hence, there is an $l \in \{1, \dots, n\}$ such that
\begin{align*}
	\rho_i &= RVaR_{\alpha_i, \beta_i}, && \beta_i>0, \quad  i=1, \dots, l\\
	\rho_i &= RVaR_{\alpha_i,\beta_i}= VaR_{\alpha_i}, && \beta_i=0, \quad  i=l+1, \dots, n.
\end{align*}

In this case we need stronger assumptions, in particular an assumption about the dependence structure of the individual risks. If these risks show some kind of negative dependence then obviously it is preferable to concentrate these risks in the reinsurer's portfolio. More interesting and maybe more realistic is the case of some kind of positive dependence. It arises for example as a result of natural catastrophes. Here we assume

\begin{description}
	\item[(A2)] The reinsurer's premium principle is consistent with  the increasing convex order, that is, $\pi(X) \leq \pi(Y)$ whenever $X \leq_{icx} Y$ for $X,Y \in L^1$.
	\item[(A3)] The random vector $X:=(X_1, \dots, X_n)$ has a PDS copula.
\end{description}	

Let $f_1, \dots, f_n \in \mathcal{C}$ be given ceded loss functions of the insurers. For all $i =1, \dots, n$ define the layer reinsurance treaty $k_{f_i}: \R_+ \to \R_+$,
\begin{align}
k_{f_i}(x) = \min \Big\lbrace \Big(x - \big(VaR_{\alpha_i+ \beta_i}(X_i) - f_i(VaR_{\alpha_i + \beta_i}(X_i)) \big)\Big)_+, \mathcal{M}_i \Big\rbrace, \label{eq:k_f_1}
\end{align}
where $\mathcal{M}_i \in [f_i(VaR_{\alpha_i+ \beta_i}(X_i)), f_i(VaR_{\alpha_i}(X_i))]$ is chosen such that $$RVaR_{\alpha_i, \beta_i}(f_i(X_i)) = RVaR_{\alpha_i, \beta_i}(k_{f_i}(X_i)).$$ Note that $k_{f_i} \in \mathcal{C}$ for $i=1, \dots, n$.

\begin{proposition}\label{thm:well-def}
	The layer reinsurance treaty $k_{f_i}$ is well-defined for all $i=1, \dots, n$. Besides, for $i > l$ it is identical with the layer reinsurance treaty $h_{f_i}$ of Section \ref{sec:VaR-min}.
\end{proposition}

\begin{proof}
	For $k_{f_i}$ to be well-defined we have to show that
	$\mathcal{M}_i \in$ $ [f_i(VaR_{\alpha_i+ \beta_i}(X_i)), f_i(VaR_{\alpha_i}(X_i))]$ exists such that \[RVaR_{\alpha_i, \beta_i}(f_i(X_i)) = RVaR_{\alpha_i, \beta_i}(k_{f_i}(X_i)).\]
	First, we consider the simpler case $i > l$. Here, $\beta_i = 0$ and hence $\mathcal{M}_i= f_i(VaR_{\alpha_i}(X_i))$.
	% Using the properties of $VaR_{\alpha}$ it follows
%		\begin{align*}
%		RVaR_{\alpha_i, \beta_i}(k_{f_i}(X_i)) &= VaR_{\alpha_i}(k_{f_i}(X_i)) = k_{f_i}\left(  VaR_{\alpha_i}(X_i) \right) \\
%				&= VaR_{\alpha_i}(f_i(X_i)) = RVaR_{\alpha_i, \beta_i}(f_i(X_i))
%	\end{align*}
%	This proves both existence of $\mathcal{M}_i$ and $k_{f_i} = h_{f_i}$.
	
	Now, let $i \leq l$. Independently of $\mathcal{M}_i $ it holds
	\[ RVaR_{\alpha_i, \beta_i}(k_{f_i}(X_i)) \leq RVaR_{\alpha_i, \beta_i}(X_i) < \infty \]
since $X_i \in L^1$. Furthermore, $\mathcal{M}_i \mapsto k_{f_i}(VaR_s(X_i))$ is a continuous function for every $s$. Hence we obtain   continuity of
	\begin{align}
		\mathcal{M}_i \mapsto RVaR_{\alpha_i, \beta_i}(k_{f_i}(X_i)) = \frac{1}{\beta_i} \int_{\alpha_i}^{\alpha_i + \beta_i} k_{f_i}(VaR_s(X_i)) ds. \label{eq:cont_in_M}
	\end{align}
	Since $f_i \in \mathcal{C}$, the monotonicity of $R_{f_i}$ implies for all $ x \geq VaR_{\alpha_i + \beta_i}(X_i)$ that
	\begin{align}
	f_i(x) &\leq x- VaR_{\alpha_i + \beta_i}(X_i)+ f_i(VaR_{\alpha_i+ \beta_i}(X_i)) \notag \\
	& = \Big(x - \big(VaR_{\alpha_i + \beta_i}(X_i) - f_i(VaR_{\alpha_i + \beta_i}(X_i)) \big)\Big)_+ \label{eq:f-k_f}
	\end{align}
	Thus for $\mathcal{M}_i = f_i \left(VaR_{\alpha_i}(X_i) \right)$ it follows from (\ref{eq:f-k_f}) and the properties of $VaR_{\alpha_i}$ that 
	\begin{align}
	RVaR_{\alpha_i , \beta_i}(f_i(X_i)) &= \frac{1}{\beta_i} \int_{\alpha_i}^{\alpha_i + \beta_i} VaR_s(f_i(X_i)) ds = \frac{1}{\beta_i} \int_{\alpha_i}^{\alpha_i + \beta_i} f_i(VaR_s(X_i)) ds \notag \\
	& \leq \frac{1}{\beta_i} \int_{\alpha_i}^{\alpha_i + \beta_i} k_{f_i}(VaR_s(X_i)) ds \notag \\
	&=  \frac{1}{\beta_i} \int_{\alpha_i}^{\alpha_i + \beta_i} VaR_s(k_{f_i}(X_i)) ds  = RVaR_{\alpha_i , \beta_i}(k_{f_i}(X_i)).\label{eq:RVaR_leq} 
	\end{align}
	Conversely, for $\mathcal{M}_i = VaR_{\alpha_i+\beta_i}(f_i(X_i))$ it holds
	\begin{align}
	RVaR_{\alpha_i, \beta_i}(k_{f_i}(X_i)) = VaR_{\alpha_i+\beta_i}(f_i(X_i)) \leq RVaR_{\alpha_i ,\beta_i}(f_i(X_i)). \label{eq:RVaR_geq}
	\end{align}
	In view of the continuity of (\ref{eq:cont_in_M}) together with (\ref{eq:RVaR_leq}) and (\ref{eq:RVaR_geq}), the intermediate value theorem ensures the existence of $\mathcal{M}_i \in [f_i(VaR_{\alpha_i+ \beta_i}(X_i)), f_i(VaR_{\alpha_i}(X_i))]$ such that the statement $RVaR_{\alpha_i, \beta_i}(f_i(X_i)) = RVaR_{\alpha_i, \beta_i}(k_{f_i}(X_i))$ follows.
\end{proof}

The reinsurance treaties $k_{f_i}$ can be used to reduce the aggregate capital requirement of all insurers.
\begin{proposition} \label{thm:RVaR_leq}
	Assume (A2), (A3) and let $f_1, \dots, f_n \in \mathcal{C}$ be arbitrary ceded loss functions. Then
	\begin{align*}
	&\sum_{i=1}^n RVaR_{\alpha_i, \beta_i} \big( R_{k_{f_i}}(X_i) \big)  + \pi \Big( \sum_{i=1}^{n} k_{f_i}(X_i) \Big)  \leq  \sum_{i=1}^n RVaR_{\alpha_i, \beta_i} \big( R_{f_i}(X_i) \big) + \pi \Big( \sum_{i=1}^{n} f_i(X_i) \Big).
	\end{align*}
\end{proposition}

\begin{proof}
First, let $i \in \{1, \dots, l\}$. The retained loss functions $R_{f_i}(x) = x - f_i(x)$ and $R_{k_{f_i}}(x)= x- k_{f_i}(x)$ are increasing continuous as $f_i, k_{f_i} \in \mathcal{C}$. Hence, by the properties of $VaR_{\alpha_i}$ and the definition of $k_{f_i}$
\begin{align}
	RVaR_{\alpha_i, \beta_i} \big( R_{f_i}(X_i) \big) &= \frac{1}{\beta_i} \int_{\alpha_i}^{\alpha_i + \beta_i} VaR_{s}(X_i ) - f_i(VaR_s(X_i)) ds \notag\\
	& = RVaR_{\alpha_i, \beta_i}(X_i) - RVaR_{\alpha_i, \beta_i}\big( f_i(X_i) \big) \notag\\
	&= RVaR_{\alpha_i, \beta_i}(X_i) - RVaR_{\alpha_i, \beta_i}\big( k_{f_i}(X_i) \big) \notag\\
	& = \frac{1}{\beta_i} \int_{\alpha_i}^{\alpha_i + \beta_i} VaR_{s}(X_i ) - k_{f_i}(VaR_s(X_i)) ds \notag\\
	&= RVaR_{\alpha_i, \beta_i} \big( R_{k_{f_i}}(X_i) \big). \label{eq:RVaR_equal}	
\end{align}
We already know from (\ref{eq:f-k_f}) that for $x \geq VaR_{\alpha_i + \beta_i}(X_i)$
\[  f_i(x) \leq \Big(x - \big(VaR_{\alpha_i + \beta_i}(X_i) - f_i(VaR_{\alpha_i + \beta_i}(X_i)) \big)\Big)_+=:\hat{k}_{f_i}(x)
.\]
Since $k_{f_i}(x)=\min\{\hat{k}_{f_i}(x), \mathcal{M}_i \}$ and
\begin{align}\label{eq:integrals_f_k_f_equal}
	\int_{\alpha_i}^{\alpha_i + \beta_i} f_i(VaR_s(X_i)) ds = \int_{\alpha_i}^{\alpha_i + \beta_i} k_{f_i}(VaR_s(X_i)) ds,
\end{align}
which holds by the choice of $\mathcal{M}_i$, there is an $x_0 \in [VaR_{\alpha_i+ \beta_i}(X_i), VaR_{\alpha_i}(X_i)]$ such that
\begin{align}\label{eq:cut}
	\begin{aligned} 
	f_i(x) &\leq k_{f_i}(x) && \text{for } x \in [VaR_{\alpha_i + \beta_i}(X_i), x_0]\\
	f_i(x) &\geq k_{f_i}(x) && \text{for } x \in (x_0, VaR_{\alpha_i }(X_i)].
	\end{aligned}
\end{align}
Now let $U \sim \mathcal{U}([\alpha_i, \alpha_i+ \beta_i])$. Then by (\ref{eq:integrals_f_k_f_equal})
\begin{align}\label{eq:eqfkf}
	\Eop[f_i(VaR_U(X_i))] = \Eop[k_{f_i}(VaR_U(X_i))].
\end{align}
This together with (\ref{eq:cut}) fulfils the assumptions of Proposition \ref{thm:stop-loss} c) and we therefore have 
\begin{align}\label{eq:VaR_icx}
	k_{f_i}(VaR_U(X_i)) \leq_{icx} f_i(VaR_U(X_i)).
\end{align}
Analogous to (\ref{eq:f-k_f})  one obtains
\begin{align} \label{eq:k_f_leq_f}
	k_{f_i}(x) \leq f_i(x) & \quad \text{for } 0 \leq  x \leq VaR_{\alpha_i + \beta_i}(X_i) \phantom{.}
\end{align}
and as $f_i$ is increasing it follows from (\ref{eq:cut})
\begin{align} \label{eq:k_f_leq_f2}
k_{f_i}(x) \leq f_i(x) & \quad \text{for } x \geq VaR_{\alpha_i }(X_i).\phantom{0 \leq_{+ \beta_i}}
\end{align}
Next, let $V \sim \mathcal{U}([0,1])$. Then $X_i \sim F^{-1}_{X_i}(1-V)=VaR_V(X_i)$ and consequently
\begin{eqnarray*}
	&&\Eop\Big[(k_{f_i}(X_i) - d)_+\Big] = \Eop\Big[(k_{f_i}(VaR_V(X_i)) - d)_+\Big]\\
	&&= \int_{[0,\alpha_i)\cup (\alpha_i + \beta_i,1]} (k_{f_i}(VaR_v(X_i)) - d)_+ dv + \int_{\alpha_i}^{\alpha_i + \beta_i}   (k_{f_i}(VaR_v(X_i)) - d)_+ dv \\
	&&= \int_{[0,\alpha_i)\cup (\alpha_i + \beta_i,1]} (k_{f_i}(VaR_v(X_i)) - d)_+ dv + \beta_i \Eop\Big[(k_{f_i}(VaR_U(X_i)) - d)_+ \Big]\\
	&& \leq \int_{[0,\alpha_i)\cup (\alpha_i + \beta_i,1]} (f_i(VaR_v(X_i)) - d)_+ dv + \beta_i \Eop\Big[(f_i(VaR_U(X_i)) - d)_+ \Big] \\
	&&= \Eop\Big[(f_i(VaR_V(X_i)) - d)_+\Big]= \Eop\Big[(f_i(X_i) - d)_+\Big],
\end{eqnarray*}
where the inequality is due to (\ref{eq:VaR_icx}), (\ref{eq:k_f_leq_f})  and (\ref{eq:k_f_leq_f2}). Hence we have $k_{f_i}(X_i) \leq_{icx} f_i(X_i)$. 

We now consider the remaining indices. Let $i \in \{l+1, \dots, n\}$. From the proof of Proposition \ref{thm:VaR_leq} together with Proposition \ref{thm:well-def} it follows that we have $k_{f_i}(X_i) = h_{f_i}(X_i) \leq_{icx} f_i(X_i)$ for these indices, too.

For all $i=1, \dots, n$ the functions $k_{f_i}$ and $f_i$ are in $\mathcal{C}$, i.e.\ increasing and continuous. Thus by (A2),  the random vectors
\[k_f(X):=\big(k_{f_1}(X_1), \dots, k_{f_n}(X_n)\big) \quad \text{and} \quad f(X):=\big(f_1(X_1), \dots, f_n(X_n)\big)\]
have the same PDS  copula as $X=(X_1, \dots, X_n)$. Hence, by Theorem \ref{thm:common_copula_icx}
\begin{align} \label{eq:sum_icx}
	\sum_{i=1}^{n} k_{f_i}(X_i) \leq_{icx} \sum_{i=1}^{n} f_i(X_i).
\end{align}
Finally, (\ref{eq:RVaR_equal}) and the fact that $\pi$ is consistent with the increasing convex order by (A3) together with (\ref{eq:sum_icx}) yields the assertion.
\end{proof}

The main result of this section is given in the next theorem. The proof follows with Proposition \ref{thm:RVaR_leq} in the same way as the proof of Theorem \ref{thm:VaR_leq}.

\begin{theorem} \label{thm:reduction}
	Under assumptions (A2), (A3), the optimal reinsurance problem (\ref{eq:ORP}) has the same optimal value as the finite dimensional problem
	\begin{align}
	%\begin{aligned}
	&\min \quad  \sum_{i=1}^n a_i + \sum_{i=1}^l RVaR_{\alpha_i , \beta_i}\big((X_i-a_i-b_i)_+\big)  + \pi \left( \sum_{i=1}^{n} \min\{(X_i-a_i)_+, b_i\}\right) \notag \\
	& s.t.  \quad 0 \leq a_i \leq VaR_{\alpha_i + \beta_i}(X_i), \quad i=1, \dots, n \notag\\
	&\phantom{\text{s.t.}} \quad a_i+b_i \geq VaR_{\alpha_i + \beta_i}(X_i), \quad i=1, \dots, l \label{eq:FORP_2}\\
	&\phantom{\text{s.t.}} \quad b_i = VaR_{\alpha_i}(X_i) - a_i , \quad i=l+1, \dots, n. \notag
	%\end{aligned}
	\end{align}
	Furthermore, if $a_1, \dots, a_n, b_1, \dots, b_n$ is an optimal solution to (\ref{eq:FORP_2}) then
	\begin{align*}
		f_i(x) &= \min\{(x-a_i)_+, b_i\}, \quad i=1, \dots,n
	\end{align*}
	is an optimal solution to (\ref{eq:ORP}).
\end{theorem}

\begin{remark}
Note that in the setting of Theorem \ref{thm:reduction} an optimal solution exists, when $\pi$ is Lipschitz continuous and for all $i=1,\dots, n$ it holds $\alpha_i>0$ or $X_i$ bounded. We have $a_i\in [0,VaR_{\alpha_i+\beta_i}(X_i)]$ and $b_i \in [0,VaR_{\alpha_i}(X_i)]$. Note that $VaR_{\alpha_i+\beta_i}(X_i) < \infty$ since we required $0<\alpha_i+\beta_i$ from the beginning and $VaR_{\alpha_i}(X_i)< \infty$ holds if and only if $\alpha_i>0$ or $X_i$ bounded. Thus the target function is minimized over a compact set. Moreover, $\pi$ is Lipschitz-continuous w.r.t.\ the supremum norm and applied to bounded random variables. Finally for monotone, translation invariant risk measures $\rho$ we obtain since $(x-c)_+ \le x_++|c|$ that $\rho((X-d)_+)-\rho((X-c)_+) \le \rho((X-c)_+ +|c-d|)_+-\rho((X-c)_+) = |c-d|$. Hence $RVaR_{\alpha_i,\beta_i}\big((X_i-a_i-b_i)_+\big) $ is Lipschitz-continuous in $a_i,b_i$ which implies the result.
\end{remark}

\section{When is the Social Optimum also Optimal for the Individual?}
When we choose $n=1$ then the optimization problem \eqref{eq:ORP} reduces to
\begin{equation}\label{eq:ORPS}
\min\quad \rho \big( R_{f}(X) \big) + \pi \big( f(X) \big) \quad  s.t.\quad  f \in \mathcal{C}.
\end{equation}
This problem can be interpreted as one where the individual insurance company is seeking for optimal reinsurance: $R_{f}(X) $ is the retained risk which is evaluated with the risk measure $\rho$ and $\pi \big( f(X) \big) $ is the premium charged by the reinsurer. In this section we will discuss when the  social optimum obtained by solving \eqref{eq:ORP} coincides with the individually optimal reinsurance treaties obtained by solving \eqref{eq:ORPS} for each single company $i=1,\ldots,n$. 

\subsection{$\pi$ is the expected value premium principle} 
An obvious case where both optimal reinsurance treaties coincide occurs when $\pi$ is the expected value premium principle i.e.\ $\pi(X) = (1+\theta) \Eop[X]$ with $\theta\ge 0$. Here we obtain $$ \pi \left( \sum_{i=1}^{n} f_i(X_i) \right)  = \sum_{i=1}^n \pi\big(f_i(X_i)\big)$$
which implies that the global optimization problem separates into local ones. Note that $\pi$ is not translation-invariant. 

\subsection{Comonotonicity together with a comonotone additive $\pi$}
Another simple case arises when the risks $X_1,\ldots, X_n$ have the upper Fr\'echet  copula and $\pi$ is comonotone additive.  Wang premium principles are for example comonotone additive (see e.g. \cite{WD98}). For a risk $X\in L^1$ they are defined as 
$$ \pi(X)= (1+\theta) \int_0^\infty g(S_X(x))  dx  $$
where $g:[0,1] \to [0,1]$ increasing with $g(0)=0$ and $g(1)=1$ is the distortion function and $\theta\ge 0$. In this setting we obtain again that $$ \pi \left( \sum_{i=1}^{n} f_i(X_i) \right)  = \sum_{i=1}^n \pi\big(f_i(X_i)\big)$$ since $f_i$ are increasing and thus preserve the copula.

\subsection{Independent risks together with exponential premium principle $\pi$}
Suppose the risks $X_1,\ldots, X_n$ are independent and $\pi$ is the exponential premium principle, i.e.
$$ \pi(X) = \frac{1}{\gamma} \log\Big( \Eop[e^{\gamma X}]\Big)$$
where $\gamma>0$ is the risk sensitivity parameter. In this setting we obtain again 
$$ \pi \left( \sum_{i=1}^{n} f_i(X_i) \right)  = \sum_{i=1}^n \pi\big(f_i(X_i)\big)$$
which implies that the global optimization problem separates into local ones. For other applications $\pi$ is also referred to as entropic risk measure.

\subsection{Individual risk measures are Value-at-Risk and $\pi$ is Lipschitz with constant $1$}
When all insurers measure their retained risk with Value-at-Risk, we know from Theorem \ref{thm:reduction_VaR} that it suffices to find an optimal solution $(a_1, \dots, a_n)$ of the finite dimensional optimization problem 
\begin{align}
	\begin{aligned}
		{\min} \quad  \sum_{i=1}^n a_i + \pi \left( \sum_{i=1}^{n} \min\{(X_i-a_i)_+, VaR_{\alpha_i}(X_i) - a_i\}\right) \\
		\text{s.t.} \quad  0 \leq a_i \leq VaR_{\alpha_i}(X_i), \quad i=1, \dots, n.
	\end{aligned}
\end{align}
in order to obtain optimal insurance treaties, which are then given by
\[f_i(x) = \min\{(x-a_i)_+, VaR_{\alpha_i}(X_i) - a_i\}, \quad i=1, \dots, n.\]  
The following proposition shows that a general solution can be obtained independently from the specific copula and marginal distributions of $X=(X_1,\dots, X_n)$ as well as the specific premium principle used by the reinsurance company.
\begin{proposition}\label{thm:VaR-solution}
	If all insurers use Value-at-Risk and $\pi$ satisfies (A1) and  is Lipschitz   with constant $1$, an optimal solution to (\ref{eq:FORP_1}) is given by
	\[ (a_1, \dots, a_n) = (0, \dots, 0) \]
	and hence the reinsurance treaties
	\[f_i(x) = \min\{x, VaR_{\alpha_i}(X_i) \}, \quad i=1, \dots, n\]
	are an optimal solution to (\ref{eq:ORP}).
\end{proposition}

\begin{proof}
Let us denote the objective function by
\[Q(a_1, \dots, a_n) =  \sum_{i=1}^n a_i + \pi \Big( \sum_{i=1}^{n} \min\{(X_i-a_i)_+, VaR_{\alpha_i}(X_i) - a_i\}\Big). \] 
We show that $Q$ is increasing which implies the assertion. First, note that the function $\min\{(x-a_i)_+, VaR_{\alpha_i}(X_i) - a_i\}$ is decreasing in $a_i$ for all $x \in [0,\infty)$, $i=1, \dots, n$. Since the premium principle $\pi$ satisfies (A1),
	\[	q(a_1, \dots, a_n) := \pi \Big( \sum_{i=1}^{n} \min\{(X_i-a_i)_+, VaR_{\alpha_i}(X_i) - a_i\}\Big)\]
	is a decreasing function. Now let $0 \leq a=(a_1, \dots,a_n) \leq (b_1, \dots, b_n) =b$ componentwise. 
By $\| \cdot\|$ we denote the supremum norm and by $\| \cdot\|_1$ the $L^1$-norm.		We have
	\begin{align}
		& Q(a_1, \dots, a_n) \leq Q(b_1, \dots,b_n) \notag \\
		\Leftrightarrow \quad & q(a_1, \dots, a_n) - q(b_1, \dots,b_n) \leq \sum_{i=1}^n b_i-a_i \notag\\ 
		\Leftrightarrow \quad & |q(a_1, \dots, a_n) - q(b_1, \dots,b_n)| \leq \|b-a\|_1,\label{eq:Q_increasing}
	\end{align}	 
	where the last equivalence holds as $q$ is decreasing and $a \leq b$. Hence, it suffices to show (\ref{eq:Q_increasing}). Let $i \in \{1, \dots, n\}$ be arbitrary. For $x \geq VaR_{\alpha_i} (X_i)$ it holds
	\begin{align}\label{eq:a-b-1}
	\begin{aligned}
	0 \leq & \min\{(x-a_i)_+, VaR_{\alpha_i}(X_i) - a_i\} - \min\{(x-b_i)_+, VaR_{\alpha_i}(X_i) - b_i\}\\
	= &  VaR_{\alpha_i}(X_i) - a_i  - VaR_{\alpha_i}(X_i) + b_i =b_i-a_i
	\end{aligned}	
	\end{align}
	and for $x < VaR_{\alpha_i}(X_i)$
	\begin{align}\label{eq:a-b-2}
	\begin{aligned}
	0 \leq &\min\{(x-a_i)_+, VaR_{\alpha_i}(X_i) - a_i\} - \min\{(x-b_i)_+, VaR_{\alpha_i}(X_i) - b_i\}\\
	=& \begin{cases}
	x-a_i-x+b_i =b_i-a_i, & b_i \leq x\\
	x-a_i-0 \leq b_i-a_i, & a_i \leq x < b_i\\
	0-0 \leq b_i-a_i, & x < a_i.
	\end{cases}
	\end{aligned}	
	\end{align}
	It follows
	\begin{align*}
		& |q(a_1, \dots, a_n) - q(b_1, \dots,b_n)|\\
		\leq \ & \left\lVert \sum_{i=1}^n  \min\{(X_i-a_i)_+, VaR_{\alpha_i}(X_i) - a_i\} - \min\{(X_i-b_i)_+, VaR_{\alpha_i}(X_i) - b_i\}\right\rVert\\
		\leq \ &  \sum_{i=1}^n \left\lVert \min\{(X_i-a_i)_+, VaR_{\alpha_i}(X_i) - a_i\} - \min\{(X_i-b_i)_+, VaR_{\alpha_i}(X_i) - b_i\}\right\rVert\\
		\leq \ & \sum_{i=1}^n  b_i - a_i = \|b-a\|_1,
	\end{align*}
	which is exactly (\ref{eq:Q_increasing}). The first inequality is due to the Lipschitz continuity of $\pi$, the second one by the triangular inequality and last one by (\ref{eq:a-b-1}), (\ref{eq:a-b-2}). 
\end{proof}
Note that every monotone and translation invariant premium principle is Lipschitz with constant $1$. Theorem \ref{thm:VaR-solution} implies that  the social optimum coincides with the individual optimum.

\subsection{A Case where Things are Different}
After so many examples showing that the social optimum coincides with the individual optimum one might wonder how examples look like where this is not the case. Since we want to discuss such a case analytically and not only numerically we make some assumptions to simplify the model.

\begin{proposition}\label{prop:sym}
Suppose that $\rho_1=\ldots =\rho_n=\rho$, $\rho$ is positive homogeneous and comonotone additive, $\pi$ is positive homogeneous and subadditive and $X_1,\ldots, X_n$ are identically distributed and have a symmetric copula. Then whenever optimal reinsurance treaties $f_1^*,\ldots ,f_n^*\in \mathcal{C}$ exist, there is a symmetric solution $ f_1^*=\ldots=f_n^*=f^*$.
\end{proposition}

\begin{proof}
Let us denote the objective function for $ f_1,\ldots ,f_n \in \mathcal{C}$ by 
$$ Q(f_1,\ldots,f_n) = \sum_{i=1}^n \rho_i \big( R_{f_i}(X_i) \big) + \pi \left( \sum_{i=1}^{n} f_i(X_i) \right)$$
and suppose that $(f_1^*,\ldots ,f_n^*)$  is an optimal solution.
We will show that then $(f^*,\ldots ,f^*)$ with  $f^*(x) := \frac{1}{n} \sum_{i=1}^n f^*_i(x)$ is also optimal. 

First it is not difficult to see that due to our assumption $Q$ is symmetric, i.e. $$Q(f_1,\ldots,f_n) = Q(f_{\sigma(1)},\ldots,f_{\sigma(n)})$$
where $\sigma$ is any permutation of the numbers $1,\ldots, n$. Thus we obtain with the positive homogeneity and subadditivity of $\pi$ and with the positive homogeneity and comonotone additivity of $\rho$ that 
\begin{eqnarray*}
Q(f_1^*,\ldots,f_n^*) &=& \frac{1}{n!} \sum_\sigma Q(f_{\sigma(1)}^*,\ldots,f_{\sigma(n)}^*)\\
&=& \frac{1}{n!} \sum_\sigma \Big(\rho(X_1-f^*_{\sigma(1)}(X_1))+\ldots+\rho(X_n-f^*_{\sigma(n)}(X_n))\\
&& + \pi(f_{\sigma(1)}^*(X_1)+\ldots +f_{\sigma(n)}^*(X_n))\Big)\\
&=&  \sum_\sigma \Big(\rho(\frac{1}{n!} X_1-\frac{1}{n!} f^*_{\sigma(1)}(X_1))+\ldots+\rho(\frac{1}{n!} X_n-\frac{1}{n!} f^*_{\sigma(n)}(X_n))\\
&& + \pi(\frac{1}{n!} f_{\sigma(1)}^*(X_1)+\ldots +\frac{1}{n!} f_{\sigma(n)}^*(X_n))\Big)\\
&\ge & \rho\big( X_1-\frac{1}{n!} \sum_\sigma f^*_{\sigma(1)}(X_1)\big)+\ldots+\rho\big(X_n-\frac{1}{n!}\sum_\sigma f^*_{\sigma(n)}(X_n)\big)\\
&& + \pi\big(\frac{1}{n!} \sum_\sigma f_{\sigma(1)}^*(X_1)+\ldots +\frac{1}{n!}\sum_\sigma f_{\sigma(n)}^*(X_n)\big)\\
&=& \rho\big( X_1- f^*(X_1)\big)+\ldots+\rho\big(X_n- f^*(X_n)\big)+ \pi\big(f^*(X_1)+\ldots + f^*(X_n)\big)\\
&=& Q(f^*,\ldots, f^*).
\end{eqnarray*}
Thus, using $f^*$ as reinsurance treaty for all insurance companies is also optimal.
\end{proof}

In order to demonstrate how the social and individual optimum can differ, we consider identically distributed, binary risks occuring for instance in term life insurance. Due to positive homogeneity of $\rho$ and $\pi$ we may assume Bernoulli distributions for all insurers. A standard way to describe the dependence structure of such risks is a Bernoulli mixture model: We assume that the individual risks can be decomposed into a common economic (or systemic) risk factor $Z \in (0,1)$ and independent idiosyncratic components, i.e. let
\[X_1, \dots, X_n \mid Z \overset{iid}{\sim} Bin(1,Z). \]
As risk measures we take $\rho_i = VaR_{\alpha}$ for $i=1, \dots, n$ and an arbitrary premium principle from the large class of Wang premium principles 
\[\pi(X)= (1+ \theta) \int_0^\infty g(S_X(x))dx, \]
with concave distortion function $g$. It holds
\[ VaR_{\alpha}(X_i) = \begin{cases}
1, &\alpha < \Eop[Z]\\
0, & \alpha \geq \Eop[Z]
\end{cases}\]
for $i=1,\dots, n$. Excluding trivial cases we assume $\alpha<\Eop[Z]$. Our example fulfills the assumptions of Proposition \ref{prop:sym}. Hence, it suffices to consider symmetric solutions to the optimal reinsurance problem $(\ref{eq:ORP})$. By Theorem \ref{thm:reduction_VaR} we have to solve
\[ \min \quad na + \pi\left( \sum_{i=1}^{n} \min\{(X_i-a)_+, 1-a\} \right) \quad \text{such that } 0 \leq a \leq 1. \]
Set $N= \sum_{i=1}^{n} X_i$. It is easy to see that $\Prob(N=k) = \binom{n}{k} \int_{0}^{1} z^k (1-z)^{n-k} dF_Z(z) =:p_{n,k} $. 
Calculating the premium principle explicitly, one obtains
\begin{align*}
	\pi\left( \sum_{i=1}^{n} \min\{(X_i-a)_+, 1-a\} \right) &= (1-a) \pi\left( \sum_{i=1}^{n} X_i \right)\\
    &= (1-a) (1+\theta) \int_0^{\infty} g(S_N(x)) dx\\
    &= (1-a) (1+ \theta) \sum_{k=1}^{n} g \left( \sum_{j=k}^{n} p_{n,j} \right)
\end{align*}
Thus, it is socially optimal that all risks are fully ceded if
\[1 > (1+\theta) \frac1n \sum_{k=1}^n g \left( \sum_{j=k}^{n} p_{n,j} \right) \]
and no reinsurance is purchased otherwise. Considering the individual optimum ($n=1$), one obtains the respective ceding condition
\[1 > (1+\theta) g \left( p_{1,1} \right) = (1+\theta) g \left( \Eop[Z] \right). \]
Since $g$ is concave we have 
\begin{align}
	\frac1n \sum_{k=1}^n g \left( \sum_{j=k}^{n} p_{n,j} \right) &\leq g \left(\frac1n \sum_{k=1}^n \sum_{j=k}^{n} p_{n,j
	} \right) \label{eq:cede}\\
	&= g \left(\frac1n \int_0^{\infty} S_N(x) dx \right) \notag\\
	&= g\left(\frac1n \Eop[N]\right) \notag \\
	&= g\left(\Eop[Z]\right). \notag
\end{align}
Whenever (\ref{eq:cede}) is strict it may happen that an individually rational insurer retains the full risk even though it would be socially optimal to cede.

\bibliographystyle{abbrv}

\end{document}